\documentclass[a4paper,11pt,oneside,english,onecolumn]{article}
\usepackage[lmargin=1in,rmargin=1in,tmargin=1in,bmargin=1in]{geometry}

\usepackage{lineno}
\usepackage{authblk}
\usepackage{amssymb}
\usepackage{amsmath}
\usepackage{amsthm}
\usepackage{color}
\usepackage{breqn}
\usepackage[english]{babel}
\usepackage[english,vlined,ruled,boxed,commentsnumbered]{algorithm2e}
\usepackage{tikz}

\newtheorem{theorem}{Theorem}
\newtheorem{lemma}[theorem]{Lemma}

\newtheorem{corollary}[theorem]{Corollary}

\theoremstyle{definition}

\newtheorem{definition}[theorem]{Definition}
\newtheorem{problem}{Problem}

\newcommand{\br}[1]{\left\{#1\right\}}                            

\newcommand{\dotProd}[2]{\ensuremath{\left\langle #1 , #2 \right\rangle }}

\newcommand{\dist}[0]{\ensuremath{\mathrm{dist}}}
\newcommand{\bigO}[1]{\ensuremath{O\left( #1 \right)}}
\newcommand{\norm}[1]{\ensuremath{\left\| #1\right\|_2}}
\newcommand{\normF}[1]{\ensuremath{\left\| #1\right\|_F}}

\DeclareMathOperator{\cost}{cost}
\DeclareMathOperator{\argmin}{argmin}

\newcommand{\ceil}[1]{\left \lceil #1 \right \rceil}
\newcommand{\append}[1]{}

\newcommand{\REAL}{\ensuremath{\mathbb{R}}}

\newcommand{\eps}{\ensuremath{\varepsilon}}

\title{Asymptotically exact streaming algorithms\thanks{Part of this work has been supported by Deutsche Forschungsgemeinschaft (DFG) within the Collaborative Research Center SFB 876 "Providing Information by Resource-Constrained Analysis",
project C4. Marc Heinrich acknowledges the support of \'Ecole Normale Sup\'erieure.}}
\author[1]{Marc Heinrich}
\affil[1]{D\'epartement d'Informatique,\authorcr
\'Ecole Normale Sup\'erieure, Paris, France\authorcr
		  \texttt{marc.heinrich@ens.fr}
}
\author[2]{Alexander Munteanu}
\author[2]{Christian Sohler}
\affil[2]{Department of Computer Science, TU Dortmund, Germany\authorcr
		  \texttt{\{alexander.munteanu,christian.sohler\}@tu-dortmund.de}
}

\begin{document}
\maketitle
\thispagestyle{empty}
\begin{abstract}
\noindent
We introduce a new computational model for data streams: asymptotically exact streaming algorithms. These algorithms have an approximation ratio that tends to one as the length of the stream goes to infinity while the memory used by the algorithm is restricted to $\operatorname{polylog}(n)$ size. Thus, the output of the algorithm is optimal in the limit. We show positive results in our model for a series of important problems that have been discussed in the streaming literature. These include computing the frequency moments, clustering problems and least squares regression. Our results also include lower bounds for problems, which have streaming algorithms in the ordinary setting but do not allow for sublinear space algorithms in our model.\\

\noindent
\end{abstract}

\section{Introduction}
Streaming algorithms aim at solving problems in a setting where the input is given as a stream of items like numerical values, points in Euclidean space or edges of a graph and moreover it is not possible to store the entire input in the main memory. Usually a streaming algorithm is allowed only one pass over the data and its working memory is restricted to polylogarithmic size in the length of the stream \cite{Muthukrishnan05}. For most non-trivial problems, it is not possible to get an exact solution with these restrictions. However, we can focus on the design of efficient approximation algorithms. The seemingly best we can hope for in this situation is a $(1\pm\eps)$-approximation. Such approximation algorithms have been developed for many interesting and important problems. Known results in this area cover a broad variety of computational problems, including $(1\pm\eps)$-approximation algorithms for estimating the frequency moments of a stream of items \cite{ApproxFreqMom}, least squares regression, low-rank approximation \cite{LinAlgStream} and clustering \cite{ClusteringMotion}. These have many applications in machine learning, classification, data mining and other fields of research.

From an information theoretic as well as statistical perspective it seems natural to say that the more data is used in a learning task, the more precise our result will be. In this context one might think of the law of large numbers or central limit theorems. However, these arguments require the observed data to follow some fixed distribution and assume for example some underlying unknown mean value or covariance structure of the data.

In a recent work \cite{MichaelJordan2013} such arguments have been leveraged to say that more data can actually lead to more efficient learning. The basic idea supporting this hypothesis is that possibly non-convex learning tasks on a ground set $S$ can be handled more efficiently by using convex relaxations $C_1\supset S$ and even further relaxations $C_l\supset\ldots\supset C_2\supset C_1 \supset S$ thereof. While this leads to considerably more efficient computations, the solution quality might decrease due to the relaxations. This is where the increase in the amount of data comes into play. Under certain distributional assumptions we can get a solution that satisfies the same bounds on the precision as approximating the initial problem without the relaxation. This is achieved by using an appropriate 
amount of additional data whose size depends on the complexity and dimension of the relaxed sets. The authors even claim a trade-off between computational cost and the amount of data used, but the technical part does not cover lower bounds to support this claim.

Our goal is to show that for many problems that have been studied in the streaming context we can also hope for improving error guarantees as we have access to more data. Firstly, our approach is distinguished from the above in that we do not impose or use statistical assumptions on the source of the data. In some cases we still have to put mild restrictions on the input stream but this is only due to the obvious fact that if some part of the result won't get enough data or only redundant data, i.e., there is no new information on that part, then we cannot hope to improve an error that we have already made. Secondly, we have lower bounds on the space complexity supporting that these assumptions are actually necessary, not only sufficient.

\paragraph*{Our model} The main purpose of the present paper is to introduce a novel model for the design of streaming algorithms. The above discussion raises the question whether we can develop streaming algorithms, which have a guarantee on the error that approaches zero as the length $n$ of the stream tends to infinity, i.e., that have an approximation ratio of $(1\pm \eps)$ for some $\eps=o(1)$, while the memory is still bounded by $\operatorname{polylog}(n)$. As the space complexity of many problems in the streaming model is polynomial in $1/\eps$ we might think of choosing $\eps=\Theta(\frac{1}{\log n})$. Given an algorithm in the usual streaming model, we could just fix $\eps$ to such a value in advance. This means that we already have non-uniform approximation algorithms in the above sense. But this requires to know the length of the stream in advance. If otherwise, the length of the stream exceeds its pre-defined limit, the algorithm will fail to satisfy the desired approximation guarantee. Our intention is therefore to develop algorithms that are \emph{uniform} in terms of $n$ and can deal with potentially infinite input streams. We will call such algorithms \emph{asymptotically exact streaming algorithms} according to the following definition.

\begin{definition}
A problem $P$ with objective function $V: S \rightarrow \mathbb{R}$ has an \emph{asymptotically exact streaming algorithm} if there exists a one-pass streaming algorithm that for an infinite input stream $I$ and every $n\in\mathbb{N}$ maintains a solution $s^{(n)}\in S$ that with probability at least $1-\delta$ satisfies $$\frac{V(s^{(n)})}{V(s^{(n)}_{opt})}\xrightarrow{n \to \infty} 1$$ where $s^{(n)}_{opt}$ is the optimal resp. exact solution to the substream of length $n$ that has been read. The space complexity of the algorithm is bounded by $\log^{\bigO{1}} n$.
\end{definition}
Asymptotically exact algorithms have been designed for several problems including bin packing \cite{BinPacking}, the traveling salesman problem \cite{TSP}, scheduling problems \cite{Scheduling} and pickup and delivery problems \cite{PDP}. However, to our knowledge, no such results exist in the streaming literature. Moreover, our approach is different, since we explicitly use the information from the input data to improve the approximation.

Our definition also gives rise to a stronger notion of approximation in the streaming context. Consider for example the problem of approximating the $k^{th}$ frequency moments $F_k$ of items that are read from an input stream. While we can give an asymptotically exact streaming algorithm for the case $k=2$, our results show that no such algorithm can exist for the case $k=0$. Note that both cases allow for $(1\pm\eps)$ streaming algorithms for any fixed $\eps>0$ \cite{ApproxFreqMom,Bar-YossefJKST02}.

\paragraph*{Our results}

We will study the following problems in the setting of asymptotically exact streaming algorithms.

\begin{problem}[\emph{$F_k$ estimation}] Let $A$ be a sequence of $n$ integers from $[N]=\br{1 \ldots N}$. The task is to compute a $(1 \pm \eps)$-estimate of the $k^{th}$ moment $F_k=\sum_{i=1}^N m_i^k, \text{ where }  m_i = | \br{j \mid A_j =i} |$
is the number of elements in the sequence that are equal to $i$.
\end{problem}

The task of estimating the frequency moments of a sequence arises in the analysis of network traffic and covers some applications on large databases and statistical data analysis. It has been shown in \cite{ImpFreqLowBound} that estimating the $k^{th}$ frequency moment requires polynomial space for $k > 2$. However, for $k\in\{0,1,2\}$, the frequency moments can be estimated using only polylogarithmic space. According algorithms and negative results were given by Alon, Matias and Szegedy in their seminal paper \cite{ApproxFreqMom} on the space complexity of estimating the frequency moments.
Using their sketching techniques, we can give an asymptotically exact streaming algorithm for maintaining an estimate on the second frequency moment of an infinite data stream within $(1\pm\eps)$ error, where $\eps=\frac{1}{\log n}$. Our algorithm uses $O(\log^2 n \log(1/\delta))$ memory words.
We also have negative results regarding the frequency moments. We are able to show that there exists no asymptotically exact streaming algorithm for estimating $F_0$, since any such algorithm must use $\Omega(n)$ space. We have similar results regarding any $k\geq 2$ if we allow for insertion and also deletion of elements.
The lower bounds on the space complexity are derived by reduction from the disjointness problem, which is known to have linear communication complexity.

Another interesting problem that has many applications in data analysis, compression and classification is the clustering problem.

\begin{problem}[\emph{clustering}] Given a set $P=\{p_1,\ldots,p_n\}\subset\REAL^d$ of $n$ points, and an integer $k$, find a set $C\subset\REAL^d$ of $k$ centers closest to the input set $P$. More precisely, the task is to minimize one of the following quantities
\begin{align*}
&\sum_{i=1}^n \dist(p_i, C)^2 		&\text{($k$-means problem)} \\
&\sum_{i=1}^n \dist(p_i, C) 	&\text{($k$-median problem)} \\
&\max_{i=1}^n \dist(p_i, C) 	&\text{($k$-center problem)}, \\
\end{align*}
where $\dist$ is the minimum distance to a set of points, i.e., $\dist(p,C) = \min_{c \in C} \dist(p,c)$.
\end{problem}

The first coreset constructions for the $k$-means and $k$-median problems are described in \cite{FirstCoreset}, \cite{HPClustering} and \cite{SmallerCoresets}. These coresets can be seen as a small size representation of the original point set such that for any choice of centers their cost is approximated up to $(1\pm\eps)$. Thus, solving the problem exactly on the small size coreset yields a solution that is within multiplicative $(1+\eps)$-error to the optimal solution. More recently, a dynamic coreset construction was designed in \cite{DynamicCoreset}. \cite{PolyCoreset} gives the first coreset construction with a polynomial dependency on the dimension, which was further improved in \cite{ImproveCoreset1} and \cite{Feldman} resulting in coresets of size $O(\frac{kd}{\eps^2})$ for $k$-median. Another construction described in \cite{TinyData} returns a coreset of size $O(\frac{k ^2}{\eps^4})$ for $k$-means clustering. Note that the size of the coreset does not depend on $d$ or $n$. The coreset is built by first projecting the input points into a lower dimensional subspace and then applying a coreset construction from \cite{Feldman}.
Our results for clustering are asymptotically exact streaming algorithms for the $k$-means as well as $k$-median problems that are based on maintaining coresets and computing the centers only on these coresets. The results are generic with respect to the coreset construction that is used. Given any such coreset construction of size $g(n, \eps )$, our algorithm constructs a summary of size $\bigO{\log n} g(n, \frac 1 {\log n} )$ and approximates the problem based on this small set to get a solution that is within $(1+\eps)$ to the optimal. 
The analysis is conducted by showing that the obtained solution from the summary is close to the optimum. We need the additional assumption that the number of points that is associated with each center increases with the number of points from the stream. This assumption is necessary as we can show by a complementing lower bound. While under this mild assumption, we can give positive results for the $k$-means and $k$-median problems, we have further lower bounds showing that no asymptotically exact streaming algorithm can exist for the $k$-center clustering problem even for $k=1$ and in $2$ dimensions.

Another problem that has been discussed extensively in the streaming literature is the least squares regression problem.
\begin{problem}[\emph{least squares regression}]
Let $A$ be an $n \times d$ matrix and $b$ be a column vector of size $n$. Find a solution $\tilde x$ such that
\begin{align*}
\norm{A\tilde x - b} \leq (1 + \eps) \norm{Ax^* - b},
\end{align*}
where $x^* = \argmin_{x \in \REAL^d} \norm{Ax - b}$ is an optimal solution.
\end{problem}

Regression is a very important problem used in machine learning and statistics to study the dependency between variables. The most efficient algorithms for solving regression problems with little time and space are due to the early works of Sarl\'os \cite{ImpLinAlg} as well as Clarkson and Woodruff \cite{LinAlgStream} and have been further optimized and generalized in the last years. Their approach is to apply space and time efficient versions of the well known Johnson-Lindenstrauss transform \cite{JLT} as a dimensionality reduction technique to reduce the space and time bounds of their algorithms.

Using the same kind of sketching technique based on random linear projections we are able to develop an asymptotically exact streaming algorithm for the least squares regression problem. On the technical part we give an asymptotically exact streaming algorithm for maintaining a sketch that allows for matrix multiplication as well as embedding a linear subspace in our setting. These problems serve as building blocks to derive the result on least squares regression as shown in \cite{ImpLinAlg,LinAlgStream}. Following the outline of these references we still have to add some additional arguments regarding the columnbasis given by the singular value decomposition (SVD) of the input matrix to account for the improving approximation guarantee. This finally enables us to derive the positive result if the input is given row-by-row and the smallest singular value of the data matrix diverges with growing $n$. Note that while the original references give algorithms in the most general model of turnstile updates \cite{Muthukrishnan05}, our input stream is much more restricted. However, on the negative side we are able to modify the lower bound arguments from \cite{LinAlgStream} to show that no asymptotically exact streaming algorithm can exists when the input matrix is given in the turnstile model even if every entry is modified only once, and
even under the additional assumption on the divergence of the smallest singular value.\\

The rest of the paper is organized as follows. In Section \ref{prelim} we give some basic definitions and repeat results that we will use in our proofs. In Section \ref{section:F2}, we describe and analyze an asymptotically exact algorithm for estimating the second frequency moment. In Section \ref{section:Clustering}, we describe our clustering algorithm for the $k$-means and $k$-median problems with improving precision. In Section \ref{section:Regression}, we describe an asymptotically exact algorithm for regression. The analysis is adapted from \cite{LinAlgStream} to work in our setting. We also derive an algorithm with improving precision for matrix multiplication as a tool for solving the regression problem. Our lower bounds can be found in the corresponding sections and are mainly derived by reduction from communication complexity problems. We conclude our paper in Section \ref{conclusion}.


\section{Preliminaries}
\label{prelim}
Before we turn to our main results regarding the introduced problems, we state some preliminary definitions and tools that we will need in our proofs.

We will assume that the error parameter $\eps$ and the failure probability $\delta$ satisfy $0 < \eps,\delta < 1/2$. For any integer $n$ we will denote by $[n]=\{1,\ldots, n\}$ the set of all integers up to $n$. For values $a,b\in\REAL$ we will write $a\in (1\pm\eps) b$ meaning that $(1-\eps) b \leq a \leq (1+\eps) b$. For any two vectors $x,y\in\REAL^d$ let $\langle x,y \rangle=\sum_{i=1}^d x_iy_i$ denote the inner product of $x$ and $y$. Throughout the paper we will consider two different matrix norms, the Frobenius norm and the operator norm.

\begin{definition}[matrix norms]
For a matrix $A \in \mathbb{R}^{n \times d}$ the Frobenius norm is defined by $\normF{A}=(\sum_{i=1}^n \sum_{j=1}^d A_{ij}^2 )^{1/2}$ and the spectral norm is given by $\norm{A}=\sup_{x\in\mathbb{R}^d\setminus\{0\}} \frac{\normF{Ax}}{\normF{x}}\;.$
\end{definition}
Note that in the special case of a vector $y\in\mathbb{R}^d=\mathbb{R}^{d\times 1}$, both matrix norms coincide with the Euclidean vector norm or length of $y$, i.e., $\normF{y}=\norm{y}=( \sum_{i=1}^d y_{i}^2)^{1/2}.$

In this paper we will make use of a sketching method described in \cite{RandomProj} that has also been used in \cite{ImpLinAlg} and \cite{LinAlgStream}. This method is based on random linear maps and is an improvement of the so called Johnson-Lindenstrauss transform from \cite{JLT}, using matrices that only consist of appropriately rescaled random entries from $\br{\pm 1}$.

\begin{theorem}
\label{th:sketching}
Fix $k = \Theta(\frac{1}{\eps^2} \log\left(\frac{1}{\delta}\right) )$. Let $R$ be a $k \times n$ matrix whose entries are independent random variables, taking values $+1$ or $-1$ with probability $1/2$ each. Let $S=\frac{1}{\sqrt k} R$. Then for an arbitrary vector $x \in \REAL^n$ we have with probability $1 - \delta$ that
$$\norm{Sx}^2 \in (1 \pm \eps)\norm{x}^2.$$
The entries on the same row of $R$ only need to be $4$-wise independent.
\end{theorem}

Using a technique from \cite{ImpLinAlg}, which consists in putting a grid of appropriate size on the unit ball and embedding the grid points, it was shown that we can have an embedding of a whole linear subspace using small size sketches. The space complexity of this task has been settled later in \cite{LinAlgStream} and \cite{NelsonN14}.
\begin{theorem}
\label{th:embedding}
Let $S$ be a $k \times n$ sketching  matrix as in Theorem \ref{th:sketching} but with $k = \Theta(\frac{d}{\eps^2}\log{\frac{1}{\delta}})$. Let $A$ be an arbitrary $n \times d$ matrix. Then with probability $1 - \delta$ we have that
\begin{align*}
\forall x \in \REAL^d: \quad \norm{SAx}^2 \in (1 \pm \eps) \norm{Ax}^2.
\end{align*}
\end{theorem}

For our lower bounds we will use some standard results from two-party one-way communication complexity. Alice and Bob are given input strings $x, y \in \br{0,1}^n$. Their goal is to compute some boolean function $f(x,y) \in \br{0,1}$ by exchanging as little information as possible. In one-way protocols, Alice sends a message and then Bob must compute the output based on this message and on its own input string. Let $R_\delta(f)$ the minimum amount of communication for a randomized two-party one-way protocol that computes $f$ with error probability at most $\delta$.

We will mainly reduce from the following two problems. In the indexing ($IND$) problem Alice is given a string $x\in\br{0,1}^n$ and Bob has an index $i \in [n]$. Bobs task is to compute $x_i$ based on Alice's message. In the disjointness problem ($DISJ$). The bit strings $x, y \in \br{0,1}^n$ received by Alice and Bob are interpreted as subsets of $[n]$ where the subset $I(a)$ contains the element $i$ if and only if $a_i=1$. Bob's task is to decide whether the two sets are disjoint. We have the following results regarding their communication complexity.

\begin{theorem}(\cite{Disjointness})
\label{th:disjointness}
$R_{1/3}(DISJ) = \Omega(n).$
\end{theorem}

\begin{theorem}(\cite{CommComplexity})
\label{th:indexing}
$R_{1/3}(IND) = \Omega(n).$
\end{theorem}


\section{Estimating the frequency moments}
\label{section:F2}

Our first problem is to maintain an estimate on the second frequency moment for an infinite stream $A=(a_i)_{i\in\mathbb{N}}$ of integers from $[N] = \{1, \ldots N\}$. Let $(e_i)_{i \leq N}$ denote the canonical basis of $\REAL^N$. We can associate the vector $e_{a_i}$ of this base with each entry $a_i$ of the input sequence. Then the frequency vector of $A$ can be expressed as $X = \sum_i e_{a_i}$ and its second frequency moment is given by $F_2=\sum_{i=1}^N x_i^2$. The second moment of the vector $X$ is equal to the squared Euclidean norm of $X$, which can be estimated within a fixed precision using the sketching matrices from Theorem \ref{th:sketching}. We will make use of this technique to develop an asymptotically exact streaming algorithm for estimating the second frequency moment.

In our first lemma we show that we will get a good estimate if we simply discard a small prefix of the input stream. More precisely, if the part that we discard is sufficiently smaller than the square root of the total size, then we will get a good estimate for the entire sequence.
\begin{lemma}
\label{lemma:F2}
Let $A$ be an input sequence of size $n$ composed as the concatenation of two parts $A^{(1)}$ and $A^{(2)}$ of size $n_1\leq \sqrt{n}$ and $n_2=n-n_1$. Let $X_2$ be the frequency vector of the subsequence $A_2$. Let $S$ be a sketching matrix according to Theorem \ref{th:sketching} with $m \geq C \frac{1}{\eps_2^2}\log{\frac{1}{\delta}}$ rows for some absolute constant $C$. Then with probability at least $1 - \delta$ we have that $\norm{SX_2}^2 \in (1 \pm \eps)\norm{X}^2,$ where $\eps = \eps_2 + \frac{3n_1}{\sqrt n}$.
\end{lemma}

\begin{proof}
First observe that using the result of Theorem \ref{th:sketching} we have with probability $1 - \delta$ that $\norm{SX_2}^2 \in (1 \pm \eps_2)\norm{X_2}^2.$ Thus, the upper bound of our claim follows immediately from $\norm{X_2} \leq \norm{X}$. For the lower bound, let $X_1$ denote the frequency vector of $A^{(1)}$ and observe that we have
\begin{align*}
\norm{SX_2} &\geq (1  - \eps_2)\norm{X_2}^2 \\
&\geq \norm{X_1 + X_2}^2 - 2\langle X_1, X_2\rangle - \norm{X_1}^2 - \eps_2\norm{X_2}^2\\ 
&=\norm{X}^2 \Big(1 - \underbrace{\frac{\norm{X_1}^2 + 2\langle X_1, X_2 \rangle + \eps_2 \norm{X_2}^2}{\norm{X}^2}}_{\eps}\Big).
\end{align*}
Now, note that $\sqrt{n} \leq \norm{X} \leq n$ and similarly $\sqrt{n_1} \leq \norm{X_1} \leq n_1$. This finally yields 
$$\eps \leq \eps_2 + \left( \frac{\norm{X_1}}{\norm{X}}\right)^2 + 2 \frac{\norm{X_1}}{\norm{X}} \leq \eps_2 + \left(\frac{n_1}{\sqrt{n}}\right)^2 + \frac{2 n_1}{\sqrt{n}} \leq \eps_2 + \frac{3 n_1}{\sqrt{n}}.$$
\end{proof}

We can already devise a quite simple algorithm from this lemma. Assume that we have already processed a first part $A_1$. We can continue to read a number of elements that is large enough to form the second part, such that the contribution of $A_1$ will be negligible. Then, we simply compute a sketch for $A_2$ with a better precision and use it as an approximation for the whole sequence. After reading enough elements, Lemma \ref{lemma:F2} ensures that we can bound the error that we make in the process.

Now, we would have a problem if we reached the point to report an estimate before the appropriate number of new elements has been read. The old sketch for $A_1$ would not meet the error bounds for the whole sequence because of the elements that have only been inserted into the second sketch and the sketch for $A_2$ would not use enough elements for the first part to be negligible. So instead, we continue to update both the first sketch and the new sketch with the elements of the sequence. If we meet the point to report before the second sketch contains enough elements, we just ignore it and return the estimate of the first sketch. This gives the following theorem.

\begin{theorem}
\label{th:F2}
There exists an asymptotically exact streaming algorithm that with probability at least $1-\delta$ maintains an estimate on the second frequency moment of a sequence within $(1 \pm \eps)$ relative error, where $\eps = \frac{1}{\log n}$. The algorithm uses $\bigO {\log^2 n \log(1/\delta)}$ memory words.
\end{theorem}

\begin{proof}
To prove the claim, we can choose $\eps_2 = \frac{1}{\log n_1}$, and $n = (\frac{n_1}{\eps_2})^2$. Then by applying Lemma \ref{lemma:F2}, we obtain an approximation of the second moment within $(1\pm\eps)$ relative error where $\eps = 4 \eps_2 \leq \frac{12}{\log n}$ because $\log n_1 \geq \log(n^{1/3}) = \frac{\log n}{3}$. Renaming $\eps$ and folding the constant factor into the memory requirements gives the required result. 
As noted in \cite{ApproxFreqMom}, the columns of the sketching matrices only need to be $4$-wise independent. Using this, each line of the matrix can be stored implicitly using only $O(\log N)$ bits of memory, i.e., one memory word. So, we only need $O(\frac{1}{\eps_i^2}\log(\frac{1}{\delta}))$ memory to store the sketching matrix at step $i$. Since each of the entries of the sketched vectors have values smaller than $n$, they can be stored in one memory word, so this is also the memory required to store the sketched vector. The memory used then follows from the fact that at any moment we only need to keep two sketching matrices.
\end{proof}

One drawback is that, in order to get a better estimate, we need to process a very large number of elements, namely more than $n_1^2$. We can reduce that by computing several sketches in parallel, at the cost of using slightly more memory. The algorithm is described in Algorithm \ref{algo:F2}. We use the bound of Lemma \ref{lemma:F2} to determine which sketch has the best error guarantee. Also note that after a certain number of steps, some of the sketches are no longer needed because some other sketch that we have started later has already reached a better error bound. In that case we can simply discard the old one instead of continuing to update it.

\begin{algorithm}
\DontPrintSemicolon 
\KwIn{A sequence $A=\{a_1, a_2, \ldots, a_n\}$ of integers, $a_i \in [N]$}
\KwOut{An approximation of $F_2$}
$i \gets 0$ \;
\While{not \emph{End of Stream}} {
  Start a new sketch with precision $\eps_i$\;
  Process the next $n_i=2^i n_0$ items for all sketches \;
  $i \gets i+1$ \;
}
\Return estimate with smallest error bound according to Lemma \ref{lemma:F2}\;
\caption{Improving algorithm for $F_2$ estimation}
\label{algo:F2}
\end{algorithm}

Compared to the previous method, where we keep only two sketches at any moment, the precision increases more often. Indeed, if one sketch is started after $n_1$ elements have been processed (these elements are dropped for this sketch), then it will require less than $(\frac{n_1}{\eps_2})^2=(n_1 \allowbreak \log n_1)^2$ more elements to become valid. The next sketch is started after $2n_1$ elements are processed, and it will require $(2n_1 \log 2 n_1)^2 \leq 8(n_1 \log n_1)^2$ more elements to become valid. So we get a better approximation every time the size of the input is multiplied by 8. However, we need to keep $\log n$ sketches at any time, so the memory used by this algorithm is increased by a factor of at most $\log n$.

\subsection{Lower bounds on estimating the frequency moments}
While we have a positive result on $F_2$-estimation, we now show that there exists no asymptotically exact streaming algorithm for $F_0$-estimation. This is particularly interesting since in the usual streaming setting both problems allow for $(1\pm \eps)$-approximation algorithms.

\begin{theorem}
\label{lem:fzero}
Any asymptotically exact streaming algorithm for estimating $F_0$ requires at least $\Omega(n)$ memory.
\end{theorem}

\begin{proof}
We reduce from the disjointness problem, which has linear communication complexity \cite{Disjointness}. Assume that we have an improving algorithm $\mathcal{A}$ to compute an approximation of $F_0$ using $s(n)$ memory. Then we create the following protocol for the disjointness problem. Alice and Bob interpret their input strings $a,b\in\{0,1\}^n$ as sets of integers $I(a) = \br{i \in [a], a_i = 1}$, resp. $I(b)$ and they have to decide whether $I(a)\cap I(b)=\emptyset$. Alice runs the algorithm on the integers from $I(a)$ and sends the memory of size at most $s(n)$ to Bob, as well as the number of elements $|I(a)|$ which can be coded in $\log n$ bits. Bob continues the execution of the algorithm, first inserting his own input set $I(b)$ once, and then inserting any element from $I(b)$ a great number of times. He will eventually reach a point where the error on the estimate $\tilde F_0$ returned by the algorithm satisfies $\tilde F_0 \in F_0 \left(1 \pm \frac{1}{2n}\right) \subseteq F_0 \pm \frac{1}{2}$, since the number of distinct elements is at most $n$. So Bob knows the exact number of distinct elements from the sequence. By comparing this quantity to $|I(a)| + |I(b)|$, Bob can decide whether the two sets are disjoint. This means that $s(n) + \log(n) = \Omega(n)$ implying $s(n) = \Omega(n)$.
\end{proof}

We can show a similar bound for $F_k$ estimation in the dynamic setting. This means in particular that our above algorithm can not be extended to work under insertions \emph{and} deletions.

\begin{theorem}
\label{lem:fk}
Any asymptotically exact streaming algorithm for estimating $F_k$, $k\geq 2$ under insertions and deletions requires at least $\Omega(n)$ memory.
\end{theorem}

\begin{proof}
The argument is very similar to the previous one. Again we reduce from the disjointness problem. Alice and Bob are given binary strings of size $n$ that they interpret as subsets of $[n]$. Then, Alice runs the algorithm on her input and communicates to Bob the memory of size $s(n)$ as well as $|I(a)|$, the number of elements in her set. Then Bob can continue the execution of the algorithm on his own input. Hereafter, he feeds a sequence consisting in repeatedly adding and removing the same element to the algorithm. Doing that a great number of times does not change the value of the frequency moments, but after some time, the error is small enough to determine $F_k$ exactly. Its value will be equal to $|I(a)| + |I(b)|$ if and only if the sets are disjoint. Thus, the number of bits exchanged by the two parties is at least $s(n) = \Omega(n)$.
\end{proof}

\section{Clustering}
\label{section:Clustering}

In this section we develop and analyze asymptotically exact streaming algorithms for the $k$-means and $k$-median clustering problems based on so called coresets. We also give lower bounds for these problems, as well as for the $k$-center problem. Let $P= (p_1, p_2, \ldots)$ be an infinite sequence of points and as previously let $P^{(n)}$ denote the first $n$ points of the sequence that acts as the input stream to our algorithm. The task is to find a set $C$ consisting of $k$ points that minimizes one of the following cost functions
\begin{align*}
&\sum_{i=1}^n \dist(p_i, C)^2 		&\text{($k$-means problem)} \\
&\sum_{i=1}^n \dist(p_i, C) 	&\text{($k$-median problem)} \\
&\max_{i=1}^n \dist(p_i, C) 	&\text{($k$-center problem)}. \\
\end{align*}
Here the distance of one point to a set of point is the minimal distance of this point to any point in the set, i.e., $\dist(p,C)=\min_{c\in C} \dist(p,c).$ 
One way to solve the above problems approximately within a fixed precision is to use coreset constructions.
\begin{definition}[\cite{HPClustering}]
Let $P$ be a set of points. A possibly weighted set of points $S$ is an $\eps$-coreset for $P$ if for any set $C$ of $k$ centers we have $\cost(C, S) \in (1 \pm \eps) \cost(C, P).$
\end{definition}
A coreset is often chosen to be a weighted subset of the original points and is typically smaller than $P$. So, a coreset acts as a summary that has approximately the same behaviour as the original point set regarding the considered problem.
The idea of a clustering algorithm based on coresets is that if a coreset contains very few points, then it is easy to compute the optimal for the coreset and then use this optimal as an approximate solution for the original point set. The property of the coreset then ensures that we have found a $(1+\eps)$-approximate solution to the original.
We now want to solve the $k$-means and $k$-median problems with a decreasing error using coresets. The algorithm is generic with regard to the actual coreset construction being used. Our negative results from Section \ref{section:negativeClustering} show that it is not possible to have asymptotically exact algorithms in the general case. However, we can still have an asymptotically exact algorithm if we impose some mild assumptions on our input point set.

In the following, we will focus on the $k$-means problem and then explain how the analysis can be adapted to work also for $k$-median. In order to simplify notations, the mean cost will be denoted by \emph{cost} omitting the subscript.

Algorithm \ref{algo:clustering} works as follows. It splits the input stream into blocks, where the block at step $i$ is of size $2^i$ and then applies a coreset construction to each of these blocks. Finally, as a summary for the whole sequence, the algorithm returns the union of the coresets obtained at each step. If the coreset construction that is used also has a failure probability, then we choose the size of the coreset in such a way that at step $i$ this failure probability is at most $\delta_i = \delta/c i^2$ for a large enough constant $c$. By the union bound, the probability that the algorithm succeeds for all steps is at least $1 - \delta$. We will denote by $P_i$ the set of points considered at step $i$, of size $n_i = 2^i$ and by $S_i$ the coreset built from $P_i$. Also, let $S = \bigcup S_i$ be the coreset for the whole point set $P^{(n)}$ after reading $n$ points and let $\eps_i = \frac{\eps_0}{\log n_i} = \frac{\eps_0}{i}$ be the error parameter at step $i$. 

\begin{algorithm}
\DontPrintSemicolon 
\KwIn{A sequence $P^{(n)}=\{p_1, p_2, \ldots, p_n \}$ of points, $p_i\in\REAL^d$}
\KwOut{A summary $S$ of these points}
$i \gets 0$\;
\While{not \emph{End of Stream}} {
  Compute a coreset $S_i$ for the next $n_i=2^i$ points with precision $\eps_i$ \;
  $i \gets i+1$\;
}
\Return $S=\bigcup S_i$\;
\caption{Algorithm for computing a coreset with improving precision}
\label{algo:clustering}
\end{algorithm}

We can apply the method of \cite{ExactKMeans} to get an optimal $k$-means clustering on the coreset in time polynomial in the size of the coreset. Now we want to show that if we have the optimal set of centers for the summary $S$, we get an approximate solution for the input point set with $(1+\eps)$ relative error for some $\eps=o(1)$ decreasing to zero as $n\to \infty$.

Given an infinite stream $P$ and an optimal solution $C^*$ to the $k$-means problem on $P^{(n)}$, we denote by $f(n)$ the minimum number of points assigned to one of the centers. More formally, if $C_1, \ldots, C_k$ are the different clusters, then $f(n)= \min_{i \in [k]} |C_i| $. We will prove the following theorem

\begin{theorem}
\label{th:improvingClustering}
Assume $f(n) \rightarrow \infty$. If $\tilde C$ is an optimal solution to $k$-means for the coreset $S$, then $\cost(\tilde C, P^{(n)}) \leq (1 + \eps) \cost(C^*, P^{(n)})$, where $\eps = O(\frac{1}{\log f(n)})$. The coreset $S$ contains at most $g(\frac{1}{\log n}, n) \log n$ points, where $g(\gamma, n)$ denotes the size of a $\gamma$-coreset on $n$ points.

\end{theorem}
\bigskip
Note, that since coresets are closed under union, we already know that $S$ is an $\eps_0$-coreset for the whole sequence. Thus we can assume without loss of generality that $\tilde C$ is already a $2$-approximation of the optimal. Before we move to the actual proof, we first use the assumption on $f(n)$ to show that $\tilde C$ and $C^*$ are \emph{close} in terms of distance as well as in terms of cost. We begin with a lemma, which shows that there is a point of $\tilde C$ near any center of $C^*$.
\begin{lemma}
\label{lemma:distanceToCenters}
We have the following inequality:
$\max_{i \in [k]} \dist(C^*_i, \tilde C)^2 \leq 12\cost(C^*, P^{(n)})/f(n).$
\end{lemma}
\begin{proof}
Let $\alpha = \frac 1 2 \max_{j \in [k]} \dist(C^*_j, \tilde C) $ and let $j_0$ be the index for which this maximum is attained. We denote by $M$ the number of points that are in the Vorono\"i cell corresponding to $C^*_{j_0}$ and at a distance of at most $\alpha$ from $C^*_{j_0}$. Since there are $f(n) - M$ points at a distance of at least $\alpha$ from $C^*_{j_0}$ inside this cell, we have $\cost(C^*, P^{(n)}) \geq \alpha^2 \left( f(n) - M \right).$ This leads to a lower bound of $M \geq f(n) - \cost(C^*, P^{(n)})/\alpha^2.$
Note that all these points are at a distance of at least $\alpha$ from $\tilde C$. We thus have $\cost(\tilde C, P^{(n)}) \geq \alpha^2 M \geq f(n) \alpha^2 - \cost(C^*, P^{(n)}).$
Now, using the fact that $\tilde C$ is a $2$-approximate solution for $P^{(n)}$, this inequality can be rewritten as $\alpha^2 \leq 3 \cost(C^*, P^{(n)})/f(n)$
which implies our claim.
\end{proof}

We now show the following lemma, which bounds the difference of cost between two sets of centers depending on their distance
\begin{lemma}
\label{lemma:ineqCostMean}
Let $C^1$ and $C^2$ two set of centers and define $\alpha = \max_i \dist(C^2_i, C^1)$. Then $\cost(C^1, P^{(n)}) \leq \cost(C^2, P^{(n)}) + n \alpha^2 + 2\alpha \sqrt{n\cost(C^2, P^{(n)})}.$
\end{lemma}

\begin{proof}
For any point $p$, let $C^2_p$ be the closest point of $C^2$ from $p$ and let $C^1_p$ be the closest point of $C^1$ to $C^2_p$. Then, applying the triangle inequality and the definition of $\alpha$ we get that
$\dist(p, C^1) \leq \dist(p, C^1_p) \leq \dist(p, C^2_p) + \alpha = \dist(p, C^2) + \alpha.$ Thus, taking the square and summing over all points, we have
$\cost(C^1,P^{(n)}) \leq \sum_{i=1}^n \left( \dist(p_i, C^2) + \alpha \right)^2 \leq \cost(C^2,P^{(n)}) + n \alpha^2 + 2\alpha\cost_{Med}(C^2, P^{(n)}).$
Note that the last term depends on the median cost instead of the mean cost.
Now, using Cauchy-Schwartz inequality the median cost satisfies
$\cost_{Med}(C^2, P^{(n)}) = \sum_i \dist(p_i, C^2) \leq \sqrt{\left( \sum_i \dist(p_i, C^2)^2 \right) \left( \sum_i 1 \right)} = \sqrt{n \cost(C^2,P^{(n)})}.$ Plugging this into the previous inequality concludes the proof.
\end{proof}

Now we have all the tools that we need and proceed with the proof of our theorem
\begin{proof}(of Theorem \ref{th:improvingClustering})
The idea is to separately analyze the points for which we have a small relative error and the others. To this end, we split the cost of the approximate into two parts. One part where we have sketches with a good enough precision and one part for which there are only few points. Then, we show for the first part, that the approximation is close enough from the optimum such that the total error that we make is small. Fix any $1/2>\eps>(\log n) ^{-1}$ let $i_0$ be the smallest index such that $\frac{1}{i_0} \leq \eps$, i.e., $i_0 = \ceil{\frac{1}{\eps}}$. For $i \geq i_0$, we have $\eps_i \leq \eps$. Now we split the cost as explained above. $$\cost(\tilde C, P^{(n)}) = \sum_{i = 1}^l \cost(\tilde C, P_i)  = \underbrace{ \cost(\tilde C,\cup_{i < i_0} P_i)}_{(*)} + \underbrace{\sum_{i \geq i_0} \cost(\tilde C, P_i)}_{(**)}.$$
Next, we bound the two terms separately. Recall that $n_i = 2^ i$ and let $\alpha = \max_i \dist(C^*_i, \tilde C)$. From  Lemma \ref{lemma:ineqCostMean} we have for the first term
\begin{align}
(*) & \leq \left(\sum_{i < i_0}  n_i\right) \alpha^2 + \cost(C^*, \cup_{i < i_0}P_i)  + 2 \alpha \sqrt{ \sum_{i < i_0} n_i  } \sqrt{\cost(C^*, \cup_{i < i_0}P_i)}  \notag \\
&\leq 2^{i_0}\alpha^2  + 2^{\frac{i_0} 2 +1}\alpha \sqrt{\cost(C^*, P^{(n)})} + \sum_{i < i_0} \cost(C^*, P_i) \notag \\
&\leq \frac{24\cdot 2^{i_0}\cost(C^*, P^{(n)})}{\sqrt{f(n)}} + \sum_{i < i_0} \cost(C^*, P_i) \label{proof:errorIneq*}
\end{align}
where the last inequality follows from Lemma \ref{lemma:distanceToCenters}. For the second term, we can apply the coreset property to each block.
\begin{align}
(**) & \leq \sum_{i \geq i_0} \frac{1}{1 - \eps} \cost(\tilde C, S_i) \leq \sum_{i \geq i_0} \cost(\tilde C, S_i) + 2\eps \cost(\tilde C, S) \notag\\
	&\leq \cost(\tilde C, S) - \sum_{i < i_0} \cost(\tilde C, S_i) + 8\eps \cost(C^*, P^{(n)}) \label{proof:errorIneq**1}
\end{align}
The last inequality is a consequence of the fact that $\tilde C$ is a $2$-approximation for $P$ and $S$ is a $1$-coreset. Thus we have $\cost(\tilde C, S) \leq 2 \cost(\tilde C, P^{(n)}) \leq 4 \cost(C^*, P^{(n)})$. Now we can leverage the fact that $\tilde C$ is an optimal solution for the set $S$ to bound the first term of the last line (\ref{proof:errorIneq**1}). Then we have
\begin{align*}
\cost(\tilde C, S) &\leq \, \cost(C^*,S) \\
	&\leq \sum_{i < i_0} \cost(C^*, S_i) + \sum_{i \geq i_0} \cost(C^*, S_i)\\
	&\leq \sum_{i < i_0} \cost(C^*, S_i) + (1 + \eps) \sum_{i \geq i_0} \cost(C^*, P_i) \\
	&\leq \sum_{i < i_0} \cost(C^*, S_i) + \sum_{i \geq i_0} \cost(C^*, P_i) + \eps \cost(C^*,P^{(n)}).
\end{align*}
Plugging this into inequality (\ref{proof:errorIneq**1}) yields
\begin{align}
(**) &\leq \sum_{i \geq i_0} \cost(C^*, P_i) + \underbrace{\cost(C^*, \cup_{i < i_ 0}S_i) - \cost(\tilde C, \cup_{i < i_0}S_i)}_{(***)} + 9\eps \cost(C^*, P^{(n)}) \label{proof:errorIneq**2}
\end{align}
Using Lemma \ref{lemma:distanceToCenters} on the middle term we get
\begin{align*}
(***) &\leq \alpha^2 \sum_{i < i_0} n_i  + 2\alpha \sqrt{\sum_{i < i_0}n_i}\sqrt{\cost(C^*, P^{(n)})} \\
&\leq \alpha^2 2^{i_0} +2^{\frac{i_0}{2}+1}\alpha \sqrt{\cost(C^*, P^{(n)})}  \leq \frac{24 \cost(C^*, P^{(n)}) 2^{i_0}}{\sqrt{f(n)}}.
\end{align*}
Putting this together with inequalities (\ref{proof:errorIneq*}) and (\ref{proof:errorIneq**2}) we have
$\cost(\tilde C, P^{(n)}) \leq \cost(C^*, P^{(n)})( 1+ 9 \eps + \frac{K \cdot 2^{i_0}}{\sqrt{f(n)}})$ for some absolute constant $K$.
Now, since we have $i_0 = \ceil{\frac{1}{\eps}} \leq \frac{1}{\eps} + 1$ if we assume that $ \eps$ satisfies $\frac{K \cdot 2^{1/\eps +1}}{\sqrt{f(n)}} \leq \eps$, then we have that $\cost(\tilde C, P^{(n)}) \leq (1 + 10 \eps) \cost(C^*, P^{(n)})$. It can easily be verified that there exists an $\eps = O(\frac{1}{\log f(n)})$ that satisfies the above condition. This concludes our proof by rescaling and renaming $\eps$.
\end{proof}

\subsection{Modifications for $k$-Median}
Now we show how the results for $k$-means can be adapted to work for the $k$-median problem. The calculations remain essentially the same up to some minor modifications. We have an inequality bounding the distance of the centers that can be derived similarly to Lemma \ref{lemma:distanceToCenters}.
\begin{lemma}
We have the following inequality: $\max_{i \in [k]} \dist(C^*_i, \tilde C) \leq 6\cost_{Med}(C^*, P^{(n)})/f(n).$
\end{lemma}

Our bound on the cost for the median case becomes stronger than Lemma \ref{lemma:ineqCostMean}, since we can remove one of the error terms, which was due to the squared distances.
\begin{lemma}
Let $C^1$ and $C^2$ two set of centers, and we note $\alpha = \max_i \dist(C^2_i, C^1)$, then we have the following inequality $\cost_{Med}(C^1, P^{(n)}) \leq n \alpha + \cost_{Med}(C^2, P^{(n)})$
\end{lemma}
Reusing the proof of Theorem \ref{th:improvingClustering} with the modified inequalities yields our result.

\begin{theorem}
Assume $f(n) \rightarrow \infty$. If $\tilde C$ is a $(1 + \eps)$-approximate solution to $k$-means for the coreset $S$, then $ \cost(\tilde C, P^{(n)}) \leq (1 + O(\eps)) \cost(C^*, P^{(n)})$
where $\eps = O(\frac{1}{\log f(n)})$. The coreset $S$ contains at most $g(\frac{1}{\log n}, n) \log n$ points, where $g(\gamma, n)$ denotes the size of a $\gamma$-coreset on $n$ points.
\end{theorem}

Contrarily to the $k$-means case, \cite{Bajaj} showed that it is not even possible to compute the $1$-median of a set of points exactly in the usual model of computation. So instead, we can use a brute-force method similar to the one used for approximating $1$-median in \cite{oneMedian} to build a grid of possible centers of size $\mathrm{poly}(n)$ and then enumerate all possible $k$ tuples from this set to get a $(1+\eps)$-approximation. By using the \emph{centroid sets} from \cite{HPClustering} one can even reduce number of possible centers to $\mathrm{polylog}(n)$ and therefore reduce the running time of the exhaustive search.

\subsection{Lower bounds on clustering problems}
\label{section:negativeClustering}
We were able to give positive results for the $k$-means and $k$-median clustering problems under the mild assumptions that the number of points in each Vorono\"i cell goes to infinity. Here we show that there is no hope for such algorithms when we drop this assumption. Furthermore, we show that in contrast to the other objectives, no asymptotically exact algorithm can exist for the $k$-center problem. We begin with the following lemma on computing the exact solution to the $k$-median and $k$-means problems.

\begin{lemma}
\label{lower:exact}
Any streaming algorithm solving $2$-means or $2$-median exactly in dimension $d=2$ uses $\Omega(n)$ memory
\end{lemma}
\begin{proof}
We reduce from the indexing problem. We denote by $a \in \br{0,1}^n$ the input string received by Alice. From $a$, Alice produces $n$ points $p_i=(1 - a_i \delta) \omega^i$ where $\omega^i$ denote the $n^{th}$ unit roots and $\delta>0$ is a constant chosen to be small enough. Alice feeds these points to the algorithm and then communicates the memory of size $s(n)$ to Bob. Given an index $j$, Bob will put points in the direction given by $\omega^j$. He puts a large number of points in the position $A$ and $B$; see Figure (\ref{figure:lowerBoundClustering}). Here the idea is that the position of one of the centers will only depend on the value $a_j$, while the other will prevent the other bits from having any influence. If the number of points at $A$ and $B$ is large enough, then the optimal centers will move close to these points. Let $c_A,c_B$ be the centers close to $A$ and $B$ respectively. Moreover, all $p_i$ with $i\neq j$ will be in the Vorono\"i cell of $c_B$. Thus, the Vorono\"i cell of $c_A$ contains all the points at $A$ and also $p_j$. Now, if $a_j=0$ then $\norm{c_A}=1$ and otherwise $\norm{c_A}<1$ since $\delta>0$. Thus, Bob can distinguish between the two cases and solve the indexing problem. Consequently $s(n)=\Omega(n)$.
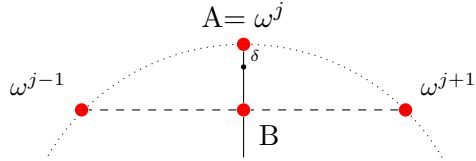
\begin{figure}
\center
\begin{tikzpicture}[scale=3]
\tikzstyle{point} = [shape=circle,inner sep =0pt, minimum size = 5, fill, red]
\tikzstyle[every node] = [draw, shape=circle, fill=red]
\node[point, label=north:{A$=\omega^j$}] (a) at (0,1) {} ;

\node[point, label=south east:B] (b) at (0,0.71) {} ;

\node[point, minimum size = 2, black] at (0,0.9) {} ;

\node[label=center:\tiny $\delta$] at (0.05,0.95) {} ;

\node[point, label=north east:\small {$\omega^{j+1}$}] (c) at (0.71,0.71)  {};

\node[point, label=north west:\small {$\omega^{j-1}$}] (d) at (-0.71,0.71) {} ;

\draw[dashed] (c) -- (b) -- (d) ;
\draw (0,0.5) -- (b) -- (a) ;
\draw[dotted] (0.87,0.5) arc [start angle = 30, end angle = 150, radius = 1];
\end{tikzpicture}
\caption{Construction for the lower bound on clustering.
}
\label{figure:lowerBoundClustering}
\end{figure}
\end{proof}

While $1$-means is trivial, in the median case, we even have the following stronger result.
\begin{lemma}
\label{lemma:1median}
Any streaming algorithm solving the $1$-median problem exactly uses $\Omega(n)$ memory.
\end{lemma}

\begin{proof}
Again, we reduce from the indexing problem. Let $a\in\br{0,1}^n$ be the string received by Alice. She produces the points $p_i = i + \frac{(-1)^{a_i}}{4}, i\in[n]$
and feeds them to the algorithm. Then, she sends the memory of size $s(n)$ to Bob. Now, we use the fact that in one dimension, if there is an odd number of points, then the median is the unique point from the input stream, such that there is an equal number of points to its right and to its left. So, by just adding the correct number of points located at $0$ and $n+1$, it is possible for Bob to select the desired point $p_j$, and thus, to retrieve the value of $a_j$. Consequently, $s(n)=\Omega(n)$ follows.
\end{proof}

Using these lemmas we can derive our negative result concerning asymptotically exact streaming algorithms.

\begin{theorem}
\label{lower:kmeanmed}
Any asymptotically exact streaming algorithm for $k$-median for $k \geq 2$ resp. $k$-means for $k \geq 3$ uses $\Omega(n)$ memory.
\end{theorem}
\begin{proof}
The idea is given an input stream of points for which we want to compute a $k$-clustering, give this stream to an asymptotically exact streaming algorithm for $(k+1)$-clustering. Then choose one point at a large enough distance from any input point and repeatedly feed the algorithm with this point. By doing this, one of the centers will move to this distant point and the other $k$ centers will provide a $k$-clustering of the points from the initial stream that does not change any more. If we repeat the insertion a large enough number of times, we will be able to retrieve the position of the centers up to an arbitrary small error implying that the solution is optimal. Now, Lemma \ref{lower:exact} and \ref{lemma:1median} yield the linear lower bound.
\end{proof}

We close this section with the lower bound on $k$-center clustering.

\begin{theorem}
\label{lower:kcenter}
Any asymptotically exact streaming algorithm for $k$-center clustering uses $\Omega(n)$ bits of space.
\end{theorem}
\begin{proof}
We show the claim already holds for the case $k=1$. 
The proof is conducted by reduction from the indexing problem. Let $a\in \{0,1\}^n$ be Alice's input and $i\in [n]$ be the index of the bit that Bob is supposed to report. Let $v_j, j\in [n]$ be the vertices of a regular $n$-polygon in clockwise order, centered at the origin and starting with $v_1=(1/2,0,\ldots)$. This construction does not depend on the choice of $a$ and $i$ but only on the size $n$. Now, suppose there is an asymptotically exact streaming algorithm for 1-center that uses space $s(n)=o(n)$. Alice inserts the vertices $v_j$ for all bits $a_j=1$ into the input stream. Then she communicates the memory of size $s(n)$ to Bob. Bob can then continue to simulate the algorithm and inserts the point $p_i=-3 \cdot v_i$. If $a_i=1$, we know that the 1-center must be located at $-v_i$ and is therefore at unit distance from $p_i$ and also from $v_i$. In particular this implies that any $(1+\varepsilon)$-approximation is at distance at least $1-\varepsilon$ from $p_i$. If on the other hand $a_i=0$, the center will appear within $(1+\varepsilon)(1-\delta)$ from $p_i$ for some fixed $\delta>0$. Bob's intention is to distinguish between these two cases. To this end, he 
 continues to insert a large number of points located at the origin. These points do not affect the optimal center in any of the cases but, since we have an asymptotically exact streaming algorithm, the approximation improves and at some point the error decreases to $\varepsilon < \delta/2$. Then we have that $(1+\varepsilon) (1-\delta) < (1+\varepsilon) (1-2\varepsilon) < (1-\varepsilon).$ This means that the possible regions for the $1$-center depending on Alice's bit are disjoint and thus Bob can recover and report the correct solution to the indexing problem. The lower bound of $s(n)=\Omega(n)$ follows.
\end{proof}

\section{Regression}
\label{section:Regression}

Next, we develop an asymptotically exact algorithm for the regression problem. Let $A$ and $b$ the input matrix and target vector obtained by processing the first $n$ items from the input stream. One way to solve this problem for a fixed precision using random sign matrices for sketching is as follows. Sketch both, the matrix $A$ and the vector $b$ with an appropriately rescaled sign matrix $S$ and solve the regression problem in the sketch space. That is, let $\tilde x$ be the optimal solution to $\min_{x \in \REAL ^d} \norm{SAx - Sb}$. It has been shown in \cite{LinAlgStream} that it is sufficient to have $\Theta(\frac{d}{\eps}\log(\frac{1}{\delta}))$ as the target dimension of the sketching matrix, such that with probability $1-\delta$ we have $\norm{A\tilde x - b} \leq (1 + \eps) \min_{x \in \REAL ^d}\norm{A x -b}$.

Now we want to compute a sketch of the matrix $A$ and the vector $b$ such that the approximate solution will have a $(1 + \eps)$ relative error, with $\eps$ decreasing to zero as the size of the input goes to infinity. We first describe how the sketching will be performed according to our algorithm and then analyze this method. In the following, we will assume that the matrix $A$ and the vector $b$ are given row-by-row. Given a matrix $D$, we will denote by $D^{(n_1)}$ the sub-matrix consisting of the first $n_1$ rows of $D$. The algorithm is given in Figure \ref{algo:sketching}. It splits the input matrix into blocks of $n_i=2^i n_0$ rows at step $i$. Then it computes a sketch of size $m_i$ for the block $i$ using rescaled sign matrices for sketching. The resulting sketch will simply be the concatenation of the single sketches at each step.

\begin{algorithm}
\DontPrintSemicolon 
\KwIn{Matrix $A\in\REAL^{n\times d}$ given row-by-row}
\KwOut{A sketch $SA$ of $A$}
$i \gets 0$\;
\While{not \emph{End of Stream}} {
  Sketch the next $n_i=2^i n_0$ rows of A with a sketching matrix $S_i$ with $m_i$ rows\;
  $i \gets i+1$\;
}
\Return concatenation of the sketches at each step\;
\caption{Improving sketching algorithm for regression}
\label{algo:sketching}
\end{algorithm}

The procedure described by the algorithm is equivalent to multiplying the input matrix $A$ and the vector $b$ by the block diagonal matrix $S = \mathrm{diag}(S_1, \ldots, S_l)$ where $S_i$ is the matrix used at step $i$ for sketching the input block. The exact number of rows $m_i$ for each of the $S_i$ will be determined later. It will be parametrized by $\eps_i$ and $\delta_i$, the error bound and failure probability of $S_i$. Now we want to prove that solving the problem for $SA$ and $Sb$ will give a good approximate solution.
In the following, we will choose $\eps_i = \frac{1}{\log n_i} = \frac{\eps_0}{i}$ and $\delta_i = \frac{\delta}{ci^2}$, where $c$ is a constant chosen to be large enough, such that $ \sum \delta_i \leq \delta$ holds. Following the outline of \cite{LinAlgStream} we begin with analyzing our sketches for some simpler problems, namely subspace approximation and matrix multiplication. We will use these results as tools or building blocks in the analysis for regression.

\subsection{Subspace approximation}
Consider the problem of subspace approximation.
\begin{problem}[\emph{Subspace approximation}]
Given a matrix $A$ the task is to find a matrix $\tilde A$ of smaller size such that for all $x \in \REAL ^d$ we have $\norm{\tilde A x }^2 \in (1 \pm  \eps) \norm{Ax}^2.$
\end{problem}

Theorem \ref{th:embedding} shows that this problem can be solved using random sign matrices for any fixed precision. Namely, if $S_0$ is a random sign matrix with an appropriate number of rows, then $S_0A$ is a solution for the subspace approximation problem with high probability. Now, it is natural to ask whether we have a  similar property for our block-diagonal matrix $S$.

We denote by $\sigma_d$ the smallest singular value of the matrix $A$.
\begin{lemma}
\label{lemma:improvedSubspaceEmbedding}
Assume that $\sigma_d \rightarrow \infty$ as the size of the input goes to infinity. Let $m_i \geq \frac{d C}{\eps_i^2} \log(\frac{1}{\delta})$ for some absolute constant $C$. Then $SA$ is a $(1 \pm \eps)$-sketch for subspace approximation with $\eps \rightarrow 0$ as $n$ goes to infinity. More precisely, if $\sigma_d^2 \geq f(n)$ for some positive, monotone function $f(n) \rightarrow \infty$, then we have $\eps = \bigO{\frac{1}{\log f(n)}}.$
\end{lemma}
\begin{proof}
By linearity, it is enough to prove the inequality for all $x$ such that $\norm x = 1$. We only show the upper bound, the lower bound can be treated similarly. Applying Theorem \ref{th:embedding} to each of the blocks, we get
$$\norm{SAx}^2 = \sum_{i=1}^l \norm{S_iA_ix}^2 \leq \sum_{i = 1}^l (1 + \eps_i) \norm{A_ix}^2  = \norm{Ax}^2\left( 1 + \frac{\sum_{i=1}^l \eps_i\norm{A_ix}^2}{\norm{Ax}^2}\right)$$
Now, since $\norm x = 1$, and since we can assume that the entries of the matrix can be stored by a logarithmic number of bits, we have that $\norm{A_ix} \leq \normF{A_i} \leq c n_i^{\Delta+1} $. So applying Lemma \ref{lemma:series} from the appendix, we have that
$$\eps = \sup_{x} \frac{\sum_{i=1}^l \eps_i\norm{A_ix}^2}{\norm{Ax}^2} \rightarrow 0.$$
In particular, using the assumption on the smallest singular value, we have that $\norm{Ax}^2 \geq \sigma_d^2 \geq f(n)$ and so there exists a constant $C$ such that
$\eps \leq \frac{C}{\log f(n)}$.
\end{proof}

 Actually, for the proof of our main theorem regarding the regression problem, we will only need that $S$ is a subspace embedding for $A$ with a fixed error $\eps_0$ (without improving precision). This is given by the following corollary. 
\begin{corollary}
\label{corollary:simpleSubspaceEmbedding}
Let $m_i \geq \frac{d C}{\eps_0^2} \log(\frac{1}{\delta})$ for some absolute constant $C$. Then $SA$ is a $(1 \pm \eps_0)$-sketch for subspace approximation.
\end{corollary}
\begin{proof}
The claim follows immediately by applying Theorem \ref{th:embedding} block-wise as in the previous proof and bounding every $\eps_i$ by $\eps_0$.
\end{proof}

\subsection{Matrix multiplication}
Before we get back to the analysis for regression, we investigate a second problem that we will need in the analysis, namely the matrix multiplication problem. Given two matrices, we want to compute their product. In the streaming setting, we can again use random sign matrices to sketch the input matrices and get a matrix that approximates the product using little space. The following result is proved in \cite{LinAlgStream}.
\begin{theorem}
\label{th:matrixMultiplication}(\cite{LinAlgStream})
Let $A$, $B$ be matrices of size $n \times d$ and $n \times d'$ respectively. If $S$ is a normalized random sign matrix of size $n \times k$ with $k \geq \frac{C}{\eps^2}\log(\frac{dd'}{\delta})$ for some constant $C$, then with probability $1 - \delta$ we have that $\normF{A^T S^T S B - A^T B} \leq \eps \normF{A} \normF{B}.$
\end{theorem}

Now we show that we can get an improving bound on the error using the same sketching method as we have used for the subspace approximation. We would like to show that we can get a $(1 + \eps)$ approximation for improving $\eps$ by sketching the matrices $A$ and $B$ with our block diagonal sketching matrix $S$. Before we prove the decreasing bound, we show the following lemma, which states that the total relative error is at most the maximum error on one of the blocks. That is, using the matrix $S$, we can get an approximation within a fixed error.

\begin{lemma}
\label{lemma:simpleMatrixMultiplication}
Let $A$ and $B$ be two matrices of size $n \times d $ and $n \times d'$ respectively. If $m_i \geq \frac{C}{\eps_1^2} \log(\frac{d d'}{\delta_i})$ then with probability $1 - \delta$ we have 
$
\normF{A^T S^T S B - A^T B} \leq \eps_1 \normF{A}\normF{B}
$
\end{lemma}

\begin{proof}
Denote by $a_i$ (resp. $b_i$) the columns of $A$ and $B$. By definition we have 
\begin{align*}
&\normF{A^T S^T S B - A^T B}^2 = \sum_{i \in [d], j\in [d']} (\dotProd{Sa_i}{Sb_j} - \dotProd{a_i}{b_j})^2
\end{align*}
Now, by the parallelogram rule \cite{ArriagaV06}, applying the sketching matrix $S$ on vectors $a_i + b_j$ and $a_i - b_j$, we have with probability $1 - \frac{2\delta}{dd'}$
$
|\dotProd{Sa_i}{Sb_j} - \dotProd{a_i}{b_j}| \leq \eps_1 \norm{a_i} \norm{b_j}.
$
So, taking the union bound over all terms of the sum yields the claim with probability $1 - 2 \delta$ since
$$
\normF{A^T S^T S B - A^T B}^2 \leq \sum_{i\in[d], j\in [d']} \eps_1^2 \norm{a_i}^2 \norm{b_j}^2  = \eps_1^2 \sum_{i=1}^{d}  \norm{a_i}^2  \sum_{j=1}^{d'} \norm{b_i}^2   = \eps_1^2  \normF{A}^2 \normF{B}^2.
$$
\end{proof}

Now, our result to get an improving bound for matrix multiplication will follow from this lemma. However, as in the case of subspace embedding, we have to impose some mild assumptions on the input.

\begin{lemma}
\label{lemma:ImprovedmatrixMultiplication}
Let $(A_n)$ and $(B_n)$ be a series of matrices of size $n \times d $ and $n \times d'$ respectively such that
$\frac{\|A_n^{(n_1)}\|_F}{\normF{A_n}} \leq \frac{cn_1^\Delta}{f(n)}$
for any $n_1$ and absolute constants $c,\Delta$. Now, for any constant $\alpha>0$ and some constant $C$, put $m_i = \frac{C}{\eps_i^{2\alpha}} \log(\frac{dd'}{\delta_i})$. Then we have that
$\normF{A_n^T S^T S B_n - A_n^T B_n} \leq \eps \normF{A}\normF{B},$
where $\eps = \bigO{\frac{1}{\log f(n)}}$.
\end{lemma}

Here $\alpha$ is an additional parameter that allows to modify the memory used depending on the convergence rate we want to achieve. To get a $(1 + \eps)$-approximate solution for regression, we will only need to apply this lemma to get a $ \sqrt \eps$-approximation for matrix multiplication, i.e., we only need $\alpha = 1/2$ removing the square in the embedding complexity. See \cite{LinAlgStream} for details.

\begin{proof}
We fix an $\eps > 0 $ to be determined later. Let $i_0$ be the smallest index such that $(\frac{1}{i_0})^\alpha \leq \eps$, i.e., $i_0 = \ceil{\frac{1}{\eps^{\frac{1}{\alpha}}}}$. By our choice, for $i \geq i_0$ we have $\eps_i^\alpha \leq \eps$. We denote by $S_{(0)}$ (resp. $S_{(1)}$) the sub-matrix of $S$ composed of the blocks with index smaller than $i_0$ (resp. greater or equal than $i_0$). Thus, we can rewrite
\begin{align*}
S = \begin{pmatrix}
S_{(0)} & \\
&	S_{(1)}
\end{pmatrix}\ A = \begin{pmatrix} A_{(0)} \\ A_{(1)} \end{pmatrix}
 \ B = \begin{pmatrix} B_{(0)} \\ B_{(1)} \end{pmatrix},
\end{align*}
where $A_{(0)}$ and $B_{(0)}$ contain $\sum_{i < i_0} n_i$ lines. Then by splitting the terms we have 
\begin{align*}
&\normF{A^T S^T S B - A^T B} \leq \normF{A_{(0)}^T S_{(0)}^T S_{(0)} B_{(0)} - A_{(0)}^T B_{(0)}} + \normF{A_{(1)}^T S_{(1)}^T S_{(1)} B_{(1)} - A_{(1)}^T B_{(1)}} 
\end{align*}

Applying Lemma \ref{lemma:simpleMatrixMultiplication} to each term separately, we have with probability $1 - \delta$ 
\begin{align}
\normF{A^T S^T S B - A^T B} &\leq \eps_1^{\alpha} \normF{A_{(0)}} \normF{B_{(0)}} + \eps \normF{A_{(1)}} \normF{B_{(1)}} \notag\\& \leq \normF{A} \normF{B} \left(\frac{\normF{A_{(0)}}}{\normF{A}} + \eps \right) \label{proof:improveMatrixMultiplication}
\end{align}
since all the blocks in $S_{(1)}$ have at least $\frac{C}{\eps^2}\log(\frac{dd'}{\delta})$ lines.
Now, using the assumption, we have that
\begin{align*}
\frac{\normF{A_{(0)}}}{\normF{A}} &\leq \frac{c(\sum_{i < i_0} n_i)^{\Delta}}{f(n)} \leq \frac{c2^{\Delta i_0}}{f(n)}.
\end{align*}

Using a similar argument as in the proof of Theorem \ref{th:improvingClustering}, we can bound the last expression by $\eps$ where $\eps = (\eps^{1/\alpha})^\alpha = \left(\frac{C'}{\log(f(n))}\right)^\alpha$ for some constant $C'$. Plugging this into inequality (\ref{proof:improveMatrixMultiplication}) yields the claim since
$\normF{A^T S^T S B - A^T B} \leq 2\eps \normF{A} \normF{B}.$
\end{proof}

\paragraph{Back to least squares regression:}
We now have all the necessary tools to prove an improving error bound for regression. In the following let $A\in\REAL^{n\times d}$ and $b\in\REAL^n$ be the input to the regression problem after the first $n$ rows have been read from the stream and let
\begin{align*}
 \tilde x &= \argmin_{x \in \REAL^d} \norm{SAx - Sb}  \\
 x^* &= \argmin_{x \in \REAL^d} \norm{Ax - b}
\end{align*}
be the optimal solution to the sketched problem and to the original problem at that point. We have the following theorem.
\begin{theorem}
\label{th:improveRegression}
Assume that the smallest singular value $\sigma_d$ of $A$ satisfies $\sigma_d^2 \geq f(n)$ for a positive monotonous function $f$ with $f(n) \rightarrow \infty$. If the blocks $S_i$ of the sketching matrix $S$ consist of $m_i \geq \frac{Cd}{\eps_i} \log(\frac{1}{\delta_i})$ rows for some absolute constant $C$, then with probability at least ${1 - \delta}$ it holds that $\norm{A\tilde x - b} \leq (1 + \eps) \norm{Ax^* - b},$
where $\eps = O(\frac{1}{\log f(n)})$.
\end{theorem}

\begin{proof}
Following the proof of Theorem 3.2 from \cite{LinAlgStream}, our claim is basically the consequence of two results. We have a subspace embedding with constant distortion using $S$. This is given by Corollary \ref{corollary:simpleSubspaceEmbedding}. Moreover, we have a result on matrix multiplication with improving precision in Lemma \ref{lemma:ImprovedmatrixMultiplication}. Let $A = U \Sigma V^T$ be the singular value decomposition of the input matrix, then Lemma \ref{lemma:ImprovedmatrixMultiplication} is applied to $U$ as one of the matrices. So it remains to justify that $U$ satisfies the assumptions of this lemma, i.e., that the norm of $U$ is not concentrated on the first $n_1$ rows. 
We denote by $u_i$ and $v_i$ the columns of $U$ and $V$ respectively and by $\sigma_i$ the $i^{th}$ singular value of $A$, which also corresponds to the $i^{th}$ diagonal coefficient of $\Sigma$. We have that $A v_i = \sigma_i u_i$. So considering only the $n_1$ first rows, we have \begin{align}
{u_i^{(n_1)}}^2 = \frac{1}{\sigma_i^2} \norm{A^{(n_1)} v_i}^2 \leq \frac{1}{\sigma_d^2} \normF{A^{(n_1)}}^2\leq c^2 n_1^{2 \Delta + 1}/\sigma_d^2. \label{proof:ineqhypmult}
\end{align}
The last inequality holds since we can assume that the entries of the matrix can be stored by a logarithmic number of bits. So we get $\norm{u_i^{(n_1)}}^2 \leq \frac{c^2n_1^{2 \Delta + 1 }}{f(n)}$ and we can rewrite $$\normF{U^{(n_1)}}^2 = \sum_{i=1}^d \norm{u_i^{(n_1)}} \leq \frac{d c^2 n_1^{2 \Delta + 1}}{f(n)}$$
Now, since $\normF{U}^2=d$, we have the required inequality and therefore all assumptions are satisfied to apply Lemma \ref{lemma:ImprovedmatrixMultiplication}.
\end{proof}

We can get the following corollary for unique updates. Here we do not enforce a row-by-row order of the streaming data, but only assume that we have at most one update per entry. Our assumptions on the smallest singular vector is still needed for the proof.

\begin{corollary}
There exists an asymptotically exact streaming algorithm to compute a $(1 + \eps)$-approximate solution to the least squares regression problem under unique updates assuming that the smallest singular value of the input matrix diverges.
\end{corollary}

The proof is almost identical to the previous theorem, so we will only describe the differences. Note that we do not get an explicit bound on the error $\eps$, but only prove that it goes to zero. The input $A$ is an $n \times d$ matrix and the number of columns $d$ is fixed at the beginning of the algorithm. Given an index $i$, we denote by $g(i)$ the largest $j$, such that the $j^{th}$ update modifies a coordinate above the $i^{th}$ row. Since there can only be a finite number of such updates, $g(i)$ is finite.
The lemma on the improving matrix multiplication is modified in the following way.
\begin{lemma}
Assume that we have a series of matrices $(A_n)$ and $(B_n)$ that satisfy
$\frac{\|A_n^{(n_1)}\|}{\normF{A_n}} \leq \frac{g(n_1)}{f(n)}$.
Then we have $\normF{A^T S^T S B - A^T B } \leq \eps \normF{A} \normF{B}$
with $\eps \rightarrow 0$.
\end{lemma}
The proof of this lemma is similar to the proof of Lemma \ref{lemma:ImprovedmatrixMultiplication} and is therefore omitted. It remains to prove that the assumptions of this lemma are still satisfied. Again, this can be done similarly to our previous proof using $\|A^{(n_1)}\|_F^2 \leq n_1 g(n_1)^{2 \Delta + 1}$ in inequality (\ref{proof:ineqhypmult}). The subsequent calculations are still correct up to immediate modifications and thus we have the required bound for the assumptions of the modified matrix multiplication lemma to hold. Corollary \ref{corollary:simpleSubspaceEmbedding} remains true without any modification.

Note, that since we do not get the entries of the matrix row-by-row any more, we now have to store the matrices $S_i$. It was shown in \cite{LinAlgStream} that the entries in the rows of the sketching matrix need only limited independence and so, these matrices can be stored implicitly using only polylogarithmic space.

\subsection{Lower bounds on regression}

We also have negative results concerning regression. These results justify the assumptions that the input is given row-by-row and that the smallest singular value diverges. We prove the following theorem.
\begin{theorem}
\label{lemma:negativeRegression}

We have the following series of results for the regression problem:
\begin{itemize}
\item Any asymptotically exact streaming algorithm for regression working under turnstile updates uses $\Omega(nd)$ memory.
\item Any asymptotically exact streaming algorithm for regression working under unique turnstile updates uses $\Omega(nd)$ memory. Here, we are allowed to modify one coordinate of the input matrix only once.
\item Any asymptotically exact streaming algorithm for regression working under turnstile updates, with the additional assumption that the smallest singular value diverges, uses $\Omega(nd)$ memory.
\end{itemize}
\end{theorem}

In particular, the third point shows that assuming only that the smallest singular value diverges is not enough. So, the other assumption, that the input is given in row-wise order is necessary.

\begin{proof}
We adapt the proof of Theorem 3.14 from \cite{LinAlgStream}. In their proof, they perform a reduction from the indexing problem with strings of size $\frac{d^2}{\eps}$ for $\eps \geq \frac{d}{36n}$. The protocol is the following. Alice feeds to the algorithm a matrix $A$ built from the input string and sends the content of the memory to Bob. Then, Bob feeds the algorithm with some modifications on $A$ and builds a vector $b$.
He retrieves the approximate optimum to the regression problem $\min_{x \in \REAL ^d}\norm{Ax-b}$ and with probability larger than $2/3$, he can guess from this approximate optimum the required bit from Alice's input string.

Now, to prove our result, we slightly modify this protocol. We denote by $A'$ the matrix and by $b'$ the vectors defined as 
$$
A' = \begin{pmatrix}
A & 0 \\ 
0 & \alpha
\end{pmatrix} 
, \ b' = \begin{pmatrix}
b \\ 0
\end{pmatrix}
$$
where $\alpha$ is just a scalar and $A$ and $b$ are given by the protocol from \cite{LinAlgStream}. Remark that the value of $\alpha$ does not modify the optimal solution, nor its value. So, for a given $\eps_0$, Alice and Bob can apply the protocol normally and then Bob simply modifies $\alpha$ repeatedly. Since we have an improving algorithm, after some time Bob will be able to get a $(1 + \eps)$-estimate of the optimal solution, where $\eps < \eps_0$. Using this optimum, Bob can retrieve the required entry from $A$. Thus, it follows that the memory used by the algorithm is at least $\Omega(d^2/\eps_0)$. Now, this lower bound holds for any $\eps_0 \geq \frac{d}{36n}$. So in particular it holds for $\eps_0$ equal to this value. This yields a lower bound of $\Omega(nd)$. The second point is obtained by just replacing $\alpha$ by a column vector. In particular, this construction proves that even under the assumption that the norm of $A$ goes to infinity, we cannot get a sub-linear improving algorithm. However, in the improving algorithm we described earlier, we made the assumption that the smallest singular value of $A$ goes to infinity. With the above constructions, this condition is not met.

To prove the third claim, we need to look a little more into the details of the protocol from \cite{LinAlgStream}. In the protocol, we have the matrix $
A = [
A_1^T,\ A_2^T,\ \cdots,\ A_k^T
]^T
$
where $A_i$ is an upper triangular matrix. Alice fills the entries of the matrix above the diagonal with $\pm 1$ according to her input string. Then Bob fills the diagonal entries with either $0$ or $P$, where $P$ is only required to be a large enough value. In particular, the diagonal entries of the matrix that contain the bit that Bob wants to retrieve are all set to $P$. Thus, updating the matrix by simply increasing the number $P$ repeatedly, the smallest singular value of $A$ goes to infinity and we would get an estimate with an error decreasing to zero. The estimated solution would still satisfy the properties needed to retrieve the bit $x_i$. So using the same argument as before, the algorithm uses $\Omega(nd)$ memory.
\end{proof}

\section{Conclusion}
\label{conclusion}
In this paper, we introduced the notion of \emph{asymptotically exact streaming algorithms}. These have an approximation ratio that tends to one as the length of the stream goes to infinity and are thus optimal in the limit. We have considered several problems from the streaming literature in this setting. Interestingly, estimating the frequency moments works in the case of $F_2$ without making additional assumptions, whereas $F_0$ does not allow for algorithms in our model. This is different in the ordinary streaming model, where both problems have $(1\pm\eps)$-approximations. However, for clustering and regression, we had to make some assumptions on the input stream to have a decreasing error bound. These were imposed to ensure that the value of a solution does not depend too much on a small number of items and were shown to be necessary. In contrast to our positive results concerning $k$-means and $k$-median clustering, there is no asymptotically exact streaming algorithm for $k$-center. It would be interesting to have similar algorithms also for other base problems like counting frequent items in a data stream and further explore the possibilities and limitations of our model. Another possible direction for future work is to extend our model to semi-streaming and graph problems.

\bibliography{bib}
\bibliographystyle{abbrv}

\section{Appendix}
\label{append}

\begin{lemma}
\label{lemma:series}
Let $(a_i)$ be a series of positive terms and $\eps_l = \frac{\sum_{i=1}^l \frac{a_i}{i}}{f(l)}$ for a positive function $f$. We assume that $f(l) \rightarrow \infty$, $f(l) \geq \sum_{i=1}^l a_i$ and $a_i \leq c2^{i \Delta}$ for some absolute constants $c,\Delta\geq 1$. Then we have that $\eps_l \xrightarrow{l \rightarrow \infty} 0$ and, more precisely, $\eps_l = \bigO{\frac{1}{\log(f(l))}}$.
\end{lemma}
\begin{proof}
Fix an $\eps > 0$ to be determined later and let $i_0$ be the smallest index such that $\frac{1}{i_0} \leq \frac{\eps}{2}$.

Then using the assumptions, we have
\begin{align*}
\eps_l &= \frac{\sum_{i=1}^{i_0 -1} \frac{a_i}{i}}{f(l)} + \frac{\sum_{i=i_0}^l \frac{a_i}{i}}{f(l)} \\
	&\leq  \frac{\sum_{i=1}^{i_0 -1} a_i}{f(l)} + \frac{\eps}{2} \frac{\sum_{i=i_0}^l a_i}{f(l)} \\
	& \leq \frac{\sum_{i=1}^{i_0 -1} a_i}{f(l)}+ \frac{\eps}{2}.
\end{align*}

It remains to bound the first term. We have
\begin{align*}
\frac{\sum_{i=1}^{i_0 -1} a_i}{f(l)} \leq \frac{c 2^{\Delta i_0}}{f(l)} \leq \frac{c 2^{\Delta \frac{3}{\eps}}}{f(l)}.
\end{align*}
If we can bound this quantity by $\frac{\eps}{2}$, then the lemma follows. It can be verified that $\eps = \frac{C}{\log f(l)}$ for some constant $C$ satisfies this condition.
\end{proof}

\end{document}